\documentclass{article}

\usepackage[preprint,nonatbib]{neurips_2020}

\usepackage[utf8]{inputenc} 
\usepackage[T1]{fontenc}    
\usepackage{hyperref}       
\usepackage{url}            
\usepackage{booktabs}       
\usepackage{amsfonts}       
\usepackage{nicefrac}       
\usepackage{microtype}      
\usepackage{xspace}
\usepackage{bbm}
\usepackage[english]{babel}
\usepackage{amsmath,amsthm}
\usepackage{tabularx}
\newtheorem{theorem}{Theorem}[section]
\newtheorem{assumption}{Assumption}[section]
\newtheorem{proposition}{Proposition}[section]
\newtheorem{corollary}{Corollary}[theorem]
\newtheorem{remark}[theorem]{Remark}
\newtheorem{lemma}[theorem]{Lemma}
\newtheorem{definition}{Definition}[section]
\usepackage{amsmath}
\usepackage{diagbox}
\usepackage{slashbox}
\usepackage{dsfont, xcolor}

\usepackage{hyperref}

\hypersetup{
    citecolor=blue,
    colorlinks=true,
    linkcolor=blue,
    filecolor=blue, urlcolor=blue,
}

\usepackage{natbib}
\usepackage{float}

\usepackage{multirow}
\usepackage[colorinlistoftodos]{todonotes}

\usepackage{subfig}
\usepackage{algorithm}
\usepackage[noend]{algpseudocode}

\usepackage{soul}
\usepackage[style=base]{caption}
\usepackage[toc,page]{appendix}
\usepackage{mathtools}

\usepackage{amssymb}

\newtheorem{example}[theorem]{Example}

\usepackage{hyperref}
\hypersetup{
    colorlinks,
    citecolor=black,
    filecolor=black,
    linkcolor=black,
    urlcolor=black
}

\usepackage[normalem]{ulem}

\title{Data Sharing Markets}

\author{%
  Mohammad Rasouli, Michael Jordan
}

\begin{document}

\maketitle

\begin{abstract}
With the growing use of distributed machine learning techniques, there is a growing need for data markets that allows agents to share data with each other. Nevertheless data has unique features that separates it from other commodities including replicability, cost of sharing, and ability to distort. We study a setup where each agent can be both buyer and seller of data. For this setup, we consider two cases: bilateral data exchange (trading data with data) and unilateral data exchange (trading data with money). We model bilateral sharing as a network formation game and show the existence of strongly stable outcome under the \emph{top agents property} by allowing \emph{limited complementarity}. We propose \emph{ordered match} algorithm which can find the stable outcome in $O(N^2)$ ($N$ is the number of agents). For the unilateral sharing, under the assumption of \emph{additive cost structure}, we construct competitive prices that can implement any social welfare maximizing outcome. Finally for this setup when agents have private information, we propose \emph{mixed-VCG} mechanism which uses \emph{zero cost data distortion} of data sharing with its \emph{isolated impact} to achieve budget balance while truthfully implementing socially optimal outcomes to the exact level of budget imbalance of standard VCG mechanisms. Mixed-VCG uses data distortions as \emph{data money} for this purpose. We further relax zero cost data distortion assumption by proposing \emph{distorted-mixed-VCG}. We also extend our model and results to data sharing via incremental inquiries and differential privacy costs.
\end{abstract}
\section{Introduction}

The recent rapid advances in machine learning algorithms, hardware, and platforms has had a major effect on fielded applications---access to data has become the limiting constraint in many such applications.  The limitation often arises because data is generated in a decentralized fashion, held privately under a variety of ownership models. There is accordingly a growing need to study distributed machine learning problems in which self-interested agents are incentivized to share their data.  Designing such incentives requires understanding the costs that users incur in sharing their data, including communication and privacy costs, and trading off those costs against possible benefits from data analysis.  The costs and benefits will generally be agent- and type-specific.  The overall problem is a complex market-design problem that aims at incentivizing the sharing of data while achieving computational, economic, and statistical efficiencies \citep{agarwal2019marketplace,  fernandez2020data, liang2020data}.

As an example consider banks who are interested in deploying up-to-date fraud detection algorithms. Each bank has its own data on financial frauds coming from its historical record of transactions.  But such data will often not have good coverage and may be out of date.  Rather than relying solely on its own data, banks would ideally share data, obtaining a better model than would be possible without sharing.  As another example, hospitals each have their own set of patient data; e.g. radiology data, genomic data, or epidemiological data. The data available to a single hospital may be not sufficient for accurate inference or particularly for causal inference.  Indeed, the per hospital data may be biased in various ways.  Ideally, hospitals would share data to mitigate these and other problems.

Data has unique features as a commodity that make data markets distinct from other markets. First, in some cases data is traded for other data (bilateral sharing) as opposed to an exchange of data for money (unilateral sharing). In bilateral sharing, the agents gain access to each others' data. Bilateral sharing is natural in use cases in which it is hard to evaluate and attach a price to data \citep{mehta2019sell}.  Indeed, the fundamental challenge of quantifying the value of data in algorithmic predictions and decision-making is an active research area \citep{ghorbani2019data,mehta2019sell}. It is also natural when regulations do not permit the exchange of money for data; this would be the case for the sharing of patient data. Bilateral trade without money also arises in other markets. Companies exchange their pool of patents with each other because they can not evaluate and attach a price to their patent pool. Also, regulatory limits on trading commodities with money also exist for various forms of matching markets, including kidney exchange markets. Relationships, such as friendships or marriage, are other forms of markets with bilateral exchange where money is not used. Even within the context of such examples, however, data markets are distinct.  In particular, data is traded often at some privacy and communication cost for agents (in contrast to patent market) and one agent can share data with multiple agents through bilateral exchanges (in contrast to organ donation or marriage markets). 

Our first problem of study in this paper, in the tradition of the seminal paper of \cite{gale1962college}, is
to find conditions under which stable outcomes exist for bilateral sharing, and can be found in polynomial time. We model bilateral sharing as a general \emph{network/graph formation game}, similar to \cite{jackson1996strategic} and \cite{jackson2005strongly}. This game is a more general version of hedonic games that are known not to have core stable outcomes in general \citep{bogomolnaia2002stability}. We also consider the notion of strong stability. In this case, without imposing a utility structure on the agents, we find conditions under which the game has a strongly stable outcome and we design an algorithm that can find such outcomes in $O(N^2)$ time, where $N$ is the number of agents. Our framework weakens some of the unrealistic assumptions that have traditionally been imposed in the literature on strongly stable outcomes \citep{ostrovsky2008stability}; in particular, we allow \emph{cyclicity} in the trades and \emph{limited complementarity}. We achieve this weakening by isolating a property---the \emph{top agent property}---which defines an ordering for agents that is based on how much each agent is preferred by all other agents. Limited complementarity then arises by ensuring that no subset of agents ranked below an agent can be preferred to that agent. Our notion of top agent property has similarities to that of \emph{ordered coalitions} in \cite{bogomolnaia2002stability}, the notion of the \emph{common ranking property} \citep{farrell1988partnerships,caskurlu2019hedonic} and the \emph{top coalition property} \citep{banerjee2001core} used in establishing the existence of core stable outcomes in hedonic games. 

We also consider unilateral data sharing where agents can assess the value of data and can trade data for money.  An example of this paradigm arises in the case of renewable energy forecasts \citep{gonccalves2020towards}. Here we build on the seminal work of \cite{shapley1971assignment}, and study conditions that allow existence of competitive prices in these markets which also implement social-welfare-maximizing sharing. Our data sharing model is a one-sided, many-to-many sharing protocol, limited to trading of a single indivisible good in each directed trade. We show that if the agents' cost of sharing is additive with respect to the agents that gain access to their data, then there exists set of competitive prices, and every social-welfare-maximizing outcome can be implemented with such competitive prices in a way that it will be a stable outcome.

Data has other unique features that influence the design of markets. In particular, the quality of data shared with other agents can be distorted upward or downward at almost zero cost, and without impact on the quality of data passed to other agents. There are several ways in which the quality of data can be distorted. For example, downward distortion can be achieved by removing part of one agent's data when passed to the another agent, and upward distortion can be obtained by adding new data that is delivered to an agent (e.g., by generating new data by using generative adversarial networks). Alternatively, noise can be added to the data, or in more complicated setups where agents share a machine learning models, the quality of those models can be distorted. All these distortions can take place without impact on the quality of data received by other agents.

Our third problem is designing market mechanisms for unilateral data sharing when agents have private types. The market is run by a social planner. We consider  Vickrey-Clarke-Groves (VCG) mechanisms \citep{vickrey1961counterspeculation}, which have been shown to uniquely implement social welfare at dominant strategies \citep{holmstrom1979groves}. However, VCG mechanisms are known to suffer from budget imbalance for the intermediary. We look for market mechanisms that have the following features: they implement approximate social-welfare-maximizing outcomes at dominant strategies, are individually rational, and are budget balanced. To this end, we propose a variation of the VCG mechanism which we refer to as \emph{mixed-VCG}. Mixed-VCG uses both data distortion, referred to as \emph{data money}, and monetary payment, to achieve budget balance while implementing social welfare via calibration to the value of budget imbalance of the standard VCG mechanism. Mixed-VCG uses downward distortion for budget surplus and upward distortion for budget deficit of the intermediary. While distorting data quality downward is practically at zero cost, for example by adding noise to the data, an upward distortion may be costly to the intermediary. To address this issue, we present a mechanism that we refer to as the \emph{distorted-mixed-VCG mechanism}. Here the intermediary adds a fixed base downward distortion of data in a zero cost way, for example fixed noise added to data. Then instead of adding new data, the intermediary simply lifts controlled amounts of the fixed noise.

While our presentation focuses on the essential components of a data sharing market, we note that our model can be extended to more complicated algorithms for data sharing. For example a popular data sharing mechanism is one in which each agent only asks for outcome of an inquiry over another agent's data and the second agent incurs a privacy cost per inquiry. We will also show that our model and results can be extended to capture this setup.

The rest of the paper is organized as follows. Section \ref{sec: lit} reviews the relevant background literature. In Section \ref{sec: bidirectional sharing}, we study our first problem for existence of stable outcomes for bilateral sharing. In Section \ref{sec: comp}, we show the existence of competitive prices for unilateral sharing, and Section \ref{sec: mech} proposes the \emph{mixed-VCG mechanism} for unilateral sharing with private information. In Section \ref{sec: dif priv} we extend our model and results to the case of data sharing limited to incremental inquiries over data under a restriction expressed in terms of differential privacy. We conclude the paper and discuss further research directions in Section \ref{sec: conclusion}.
 
\section{Related Work} \label{sec: lit}

In this section we review some of the relevant literature, focusing on four main areas: an emerging literature on data markets, and three classical areas in microeconomics, specifically stable outcomes in cooperative games, competitive equilibrium, and social-welfare-maximizing mechanism design. We discuss these areas separately below.

\underline{Data markets and relevant applications:} There is a growing recent literature for data markets that studies specific properties of data as a economic good [see, e.g., for a survey \cite{mehta2019sell}]. The literature addresses market structures and mechanisms for sales of data and data products (e.g., ratings, predictions, recommendations, and machine learning models). 

To the best of our knowledge, our paper is the first to study non-monetary mechanisms for the trade of data (bilateral data exchange). There has been, however, a growing literature on monetary data market mechanisms which we review here. The first set of papers in this area study data sharing among strategic agents/firms. These papers implement information sharing among competing agents via the use of an intermediary. In this literature the intermediary is treated as a welfare maximizer, and is limited to aggregating data from agents and sharing it with all others~\citep[see, e.g.,][]{raith1996general}. \cite{bergemann2013robust} and \cite{bergemann2018design} extend this literature to the case of a rent-seeking intermediary who can also distort data.  In contrast to this literature, we do not consider competition among the agents in downstream markets, and we also do not pose corresponding restricting assumptions on trade and on utility functions. We allow for bilateral trade, both with and without intermediaries. In our study, the intermediary is a social-welfare maximizer and we also allow the intermediary to add noise to the model.

The other relevant strand of the data-market literature are two-sided markets where buyers and sellers are separate. \cite{agarwal2019marketplace} propose a market mechanism for finding prices and distributing welfare while allowing individual prices per buyer.  The prices are not allowed to be seller specific. They also not consider the cost of sharing, including privacy costs. Our work extends this model to buyer-seller specific prices in a setup where buyers can also be sellers. We also consider the cost of sharing, as do \cite{ghosh2011selling} and \cite{ligett2012take} who study the tradeoff with revenue gain.  \cite{acemoglu2019too} discuss the externalities in data privacy. 

There are several applications similar to data sharing. These can be unified under \emph{information economics}. Patents have similar features to data  \citep{lerner2004efficient}. They can be sold to multiple agents, can be combined together for higher productivity, and can be ordered in their value for consumers. However, the literature on patents assumes a negligible cost of sharing. Moreover patent quality can not be distorted by the intermediary. Ad auctions \citep{varian2009online} and online matching \citep{mehta2013online} are other relevant applications where intermediaries sell data about customers or the other side of the market. In these markets the sales of data is limited, a single ad shown to a target customer, data for each customer is sold to a single advertiser, and there is a one-to-one matching in online matching.

\underline{Stable outcomes in cooperative games:} There is an extensive literature of cooperative matching games dating back to the seminal work of \citep{gale1962college}. The major question is to existence of stable outcomes and efficient (polynomial order) algorithms for finding them. The games studied differ in various ways---in the structure of the matching market (e.g., two-sided vs one-sided, one-to-one vs one-to many vs many-to-many, etc.), and in the particular notion of stability employed (strong stability, core stability, Nash stability, individual stability, contractual individual stability, etc). \citep{gale1962college} studied two-sided games with two-sided preferences and established the existence of a core stable outcome via a deferred acceptance algorithm that runs in time $O(n^2)$. Two-sided matching games with one-sided preferences and with ownership is studied in \cite{shapley1974cores} where the top-trading cycle algorithm demonstrates the existence of a core stable outcome calculated in time $O(n^2)$.  Finally, two-sided games without ownership are studied in \cite{hylland1979efficient}, who studied a class of ``random serial dictatorship'' algorithms, and in \cite{abdulkadirouglu1998random}, who proposed a ``core from random endowment mechanism.''  While the majority of this work involves two-sided matching, \cite{irving1985efficient} studied core stable outcomes in one-sided team formation problems in the context of roommate matching. This is a special case of hedonic games where agents form disjoint partitions and have preferences only in terms of agents in their own coalition \citep{bogomolnaia2002stability}. Hedonic games do not have core stable outcomes in general.
\cite{bogomolnaia2002stability} show that hedonic game with ordered coalitions or when players have additive and separable preferences have individually stable coalitions.

Our game is more general than hedonic games since matchings are not necessarily in coalition partitions; it could be that agents A and B share data together but C only shares data with A and not B. Therefore our game is a \emph{network/graph formation game} \citep{jackson1996strategic}. 

Another difference between our work and existing literature is that the latter focuses on core stable outcomes, while focus on strong stability.  We will henceforth refer to strong stability as ``stability.'' Recall that the outcome of the game is strongly stable if there is no deviating coalition of agents that can form contracts among each other, also keeping or dropping their existing contracts with agents out of the deviating coalition at will, in a way such that all members of the deviating coalition weakly prefer the new outcome and at least one of them strongly prefers the new outcome. 

Most of the literature is based on two assumptions that are used to establish the existence of strong stability  \citep{ostrovsky2008stability}: substitutability and acyclicity of trade. We relax  both assumptions in this study.  

Games that weaken the acyclicity assumption are studied by \cite{jackson1996strategic} in the context of friendship networks.  The results, however, are limited to a weak notion of pair-wise stability and to restricted value allocation for the outcomes. \cite{jackson2005strongly} studied strong stability in network formation games but imposed strong assumptions on the structure of the value function; for example, component-wise separability. In contrast we consider strong stability and general preferences without value functions. 

A line of work that relaxes complementarity is presented for large markets by \cite{azevedo2018existence} but it relies on the existence of transferable (monetary payments) between agents. \cite{pycia2012stability} study the existence of stability in two-sided matching markets but assume that preferences are pair-wise aligned in two sides. In this paper, we study a network formation game involving finite number of agents with  general preferences (without utility functions). We also consider a \emph{limited complementarity property} which is satisfied (for example) under the substitutability condition of \cite{ostrovsky2008stability} but also holds more broadly.

In our work the limited complementarity property arises from the \emph{top agent property}, which induces an ordering among the agents on how much they are preferred by others. Our notion of top agent property has similarities to that of ``ordered coalitions'' in \cite{bogomolnaia2002stability} as well as the notion of the ``common ranking property''  \citep{farrell1988partnerships,caskurlu2019hedonic} and the ``top coalition property''~\citep{banerjee2001core} used for establishing the existence of core stable outcomes in Hedonic games. 

Finally while it is NP hard to find core stable outcomes in hedonic games under the assumptions that are generally made in the literature, in our setup we we are able to derive an algorithm that finds stable outcomes in time $O(N^2)$, where $N$ is the number of agents.

\underline{Competitive equilibrium:} There is an extensive literature on competitive prices. The main questions are existence of competitive prices, their structure, mechanisms that can find/form those prices, and their properties (including fairness, stability, and efficiency).  \cite{shapley1971assignment} studied two-sided matching games in which agents can have monetary transfer and show the existence of competitive prices. They also show that these prices realize core stable outcomes, and are social-welfare maximizing. The extension to the case of non-quasilinear utility is studied in \cite{demange1985strategy} and \cite{alaei2016competitive}.  \cite{sotomayor1992multiple} provide an extension to multi-partner two-sided matching markets case where each worker (firm) can have multiple but a limited number of jobs (hiring positions). Further study of many-to-one matching and many-to-many matching in two-sided matching markets is provided in \cite{gul1999walrasian} and \cite{ausubel2002ascending}. Our data sharing market is not two-sided and allows multiple partners without limitation on their number. we also consider quasilinear utilities. However, we only allow one unit of an indivisible good to be traded in each directed partnership.  Considering an additive structure on the utility function, we establish the existence of competitive prices for this setup, and we show that every social-welfare-maximizing outcome can be implemented with competitive prices such that it will be a stable outcome.

\underline{Social-welfare-maximizing mechanism design:} There is large literature on mechanism design where agents have private information. The VCG mechanism, presented in the seminal paper of \cite{vickrey1961counterspeculation}, is the only truthful dominant mechanisms that maximizes the utilitarian social welfare under certain assumptions, including unrestricted preferences and at least three different outcomes \citep{holmstrom1979groves}. VCG mechanisms are also ex-post individually rational under monotonicity of the choice set and non-negativity of externalities. However, VCG mechanisms have major practical shortcomings. In particular, they are not individually rational and budget balanced at the same time. 

Green-Laffont's impossibility result shows that no dominant-strategy incentive-compatible mechanism is always both social-welfare maximizing and weakly budget balanced \citep{green1977characterization}. We extend the VCG mechanism to address these issues by leveraging the structure of the data sharing economy. We consider unique features of data, particularly the possibility of distort the quality of data for an agent at zero cost and without impact on other agents' utilities. We then introduce a ``mixed-VCG'' framework that achieves approximately optimal social welfare, dominant strategy incentive compatibility, budget balance, and ex-post individual rationality. 

\section{Bidirectional Data Sharing: The Network Formation Game}\label{sec: bidirectional sharing}

We begin by presenting our model of the bidirectional data sharing problem, posing it as as a network formation game. Then we present our results on the existence of strongly stable outcomes for this game.

\subsection{Model for bidirectional data sharing}\label{sec: model}

Consider $N$ agents (e.g., banks or hospitals). Each agent generates value from data available to them (e.g., by training predictive models for fraud or cancer detection). Agent $n$ has a local data batch $D_n$ of size $d_n=|D_n|$. 
Agents can also share data with each other, but sharing comes at some cost (e.g., a privacy cost or communication cost).
\begin{assumption}\label{ass: full sharing}
Upon sharing, data is reported fully and truthfully.
\end{assumption}
The truthful report of data can be justified practically by the existence of an auditing third party. Furthermore, data sharing can be unidirectional (data in exchange of money) or bidirectional (data in exchange of data). 
\begin{assumption}(Bidirectional Sharing) Once agents $n$ and $m$ share data then both gain access to the others' data.
\end{assumption}
Bidirectional data sharing is in the cases where it is hard to attach a monetary value to data (e.g. for banks rare fraud data) and hence data is exchanged with data, or the case that regulations prevent sales of data for money (e.g. patients data for banks).

We model the data sharing game as a \emph{Network Formation Game}, based with a graph $G$ that is defined as follows.
\begin{definition}[Sharing Graph]\label{def: sharing graph}
In the sharing graph $G$ each node is an agent and each edge represents the sharing of data between two nodes. For bidirectional sharing, an edge between nodes $n$ and $m$ means they share data in both directions (we exclude edges from a node to itself). 
\end{definition}

If agent $n$ weakly prefers outcome graph $G_1$ to $G_2$ we denote this preference by $G_1\succcurlyeq_i G_2$. Strong preference is denoted by $G_1\succ_i G_2$.

Our work is based on the following notion of stability.
\begin{definition}[Stable Outcome] A graph $G$ is stable if no coalition of agents can form a new subgraph while retaining some of their existing edges, such that all members of the coalition weakly prefer the new graph to the original graph, and at least one of them strongly prefers it. 
\end{definition}

Next, we provide structures on agents' preferences. Denote by $S_n$ the set of nodes connected to $n$ including $n$ itself.

\begin{assumption}[Preference Structure]\label{ass: pref structure 1} Agent $n$'s preferences over outcome graphs in $\mathcal{G}$ are a function solely of $S_n$. 
\end{assumption}

\begin{remark}
The structure in Assumption \ref{ass: pref structure 1} does not allow for the modeling competition of agents in downstream markets, because each agent's preference is independent of what data is available to himself and not the other agents.
\end{remark}

The preference structure in Assumption \ref{ass: pref structure 1} is a common assumption in the literature on network formation games and hedonic games.

The graph formation game above is similar to that of \cite{jackson1996strategic} and \cite{jackson2005strongly}, with the distinction that in our case agents do not have a utility function, which makes the game more general.

\subsection{Existence of stable outcomes}
The network formation game does not have a stable outcome in general. We present the \emph{top agent} and \emph{limited complimentarity} assumptions and prove under these assumptions, stable outcomes always exist. We also show by example that relaxing any of these assumptions leads to no stable outcome.

\begin{definition}[Total Preference Ordering]\label{def: comp order} We say agent $k$ prefers $i$ to $j$, denoted by $i \succ_k j $, if for any subset of agents $S_k$ with $\{i, j\}\cap S_k=\emptyset$, $S_k\cup \{i\} \succ_k S_k\cup \{j\}$. Agent $k$ has a total preference ordering if it has preference over every other two agents and furthermore this preference is complete, transitive, and reflexive.
\end{definition}

\begin{definition}[Top Agent Property]\label{def: top agent}
Agents preferences satisfy the top agent property if all agents have a total preference ordering, and furthermore their preferences over other agents are the same.
\end{definition} 

Under the top agent property, the agents can be ranked from the top choice agent to the lowest preferred agent for data sharing.

\begin{corollary}[Homogeneous Preference Structure]\label{cor: utility structure 2}
In data sharing markets, if each agent $n$'s preference is only a function of the total size of shared data that includes him, $\sum_{k\in S_n} d_k$, and agents' data sizes are orderable, $d_1>d_2>...>d_N$, then the top agent property holds.
\end{corollary}

\begin{definition}[Limited Complementarity] \label{def: limited complementrity} Agents have limited complementarity if for every $k, i$ such that $S_k\succ_k S_k\cup \{i\}$ then $S_k \succ_k S_k\cup S'$ for all $S'\subset \{i, i+1, ...., N\}$.
\end{definition}

Limited complementarity argues that if adding an agent $i$ to the set $S_k$ decreases the preference for agent $k$, then adding any subset of agents with less than or equal preference to $k$ will also decrease the preference for agent $k$; in other words, the complementarity for agent $k$ across any such subset is limited from above.

We are now ready to establish the existence of stable outcomes for network formation under the top agent and limited complementarity assumptions.

\begin{proposition}\label{prop: core} Under top agent assumption in Definition \ref{def: comp order} and limited complementarity assumption in Definition \ref{def: limited complementrity}, there exists a stable outcome that can be found in time $O(N^2)$ for the network formation game.
\end{proposition}
\begin{proof}
We prove by constructing an algorithm. Since total preference ordering exists, without loss of generality we assume agents are sorted in their value for others, with $1$ the highest preferred and $N$ the least preferred. Consider the following algorithm which we refer to by the \emph{ordered match}. Initially set $G$ to be an empty edge graph. Agent $1$ moves first (set $n=1$), swipes $\ell$ from $2$ to $N$ and proposes to each agent $\ell$ to share data together if
\begin{align}
    S_n\cup\{\ell\}\succcurlyeq_n S_n\label{eq: cond 1 join},
\end{align} 
Agent $\ell$ accepts if 
\begin{align}
    S_\ell\cup\{n\}\succcurlyeq_\ell S_\ell.\label{eq: cond 2 join}
\end{align}
If $\ell$ accepts then an edge between $n=1$ and $\ell$ is added to $G$. $S_n$ and $S_\ell$ are also updated correspondingly. Then $1$ moves to $\ell+1$. After $1$ finishes his swipe, agent $2$ does the same process (set $n=2$) but passing through $3$ to $N$ and so on. We denote the outcome graph of the ordered algorithm run over the set of agents $S$ by $Ord(S)$.

To show the outcomes of the ordered match are stable, we use the following lemma and corollary.

\begin{lemma}\label{lemma: ord match remove edge} All agents weakly prefer the ordered match outcome graph to the alternative graph formed from removing some of its edges.
\end{lemma}
\begin{proof}
Denote the outcome of the ordered match by $G$. Consider removing some edges from $G$ resulting in graph $G'$. Denote the set obtained from removing exactly $i$ of the least valuable nodes from $S_m$ by $S_m^{-i}$.  It is obvious from the ordered match algorithm that $S_m\succcurlyeq_m S_m^{-i}$. Consider a node $m$ with the set of removed connections $R'_m:=S_m\backslash S'_m$, where $|R'_m|$ is the size of $R'_m$. To complete the proof, we show that $S'_m\succcurlyeq_m S_m^{-|R'_m|}$. To show this, we construct an algorithm that starts from $S'_m$ and converges to $S_m^{-|R'_m|}$ by replacing the nodes such that in each step the preference of agent $m$ is non-decreasing. The replacement algorithm is in two steps which are repeated for $|R'_m|$ times: (a) first extend the set $S'_m$ by adding the least valuable node in $R'_m$ to it to form a new set denoted by $S'_{m, ext}$, (b) remove the least valuable node in $S_m$ from $S'_{m,ext}$ to form $S''_m$. Note that this least valuable node of $S_m$ always exists in $S'_{m,ext}$, because either it has been a member of $S'_m$ from the beginning, or otherwise it is added in step $(a)$. By the total preference ordering assumption, Definition \ref{def: comp order}, we have $S''_m\succcurlyeq_m S'_m$. We next run steps (a) and (b) on $S''_m$ by adding the second least valuable node in $R_m$ to $S''_m$ and removing the second least valuable node of $S_m$ from $S''_m$. Again agent $m$ has non-decreasing preferences over the resulting set. Continuing this procedure $|R_m|$ times forms set $S_m^{-|R_m|}$.
\end{proof}

\begin{corollary}\label{coro: augmenting} If agent $i$ does not offer to agent $j$ in ordered match algorithm---i.e., (\ref{eq: cond 1 join}) does not hold---then adding any subset of agents in $\{j,j+1,...,N\}$ to $S_i$ strongly decreases $i$'s utility.
\end{corollary}
This corollary is a direct consequence of the limited complementarity assumption in Definition \ref{def: limited complementrity}.

We now complete the proof of the theorem. Consider $G=Ord(S)$. Consider a deviating coalition $S'=\{1',2',...,N'\}$ of agents with agents ordered in their value with $1'$ being the highest valued agent. Denote their deviating sharing graph, which consists only of nodes $S'=\{1',2',...,N'\}$, by $G'$. Note that if $G'\subset G$ then no agent is strongly better with the deviation due to lemma \ref{lemma: ord match remove edge}, so $G'\not\subset G$. Next, search for a new edge in $G'$ compared to $G$ using the following \emph{ordered search} algorithm: start from $1'$ and swipe over $2', 3',...$ to the first edge in the coalition graph $G'$ that did not exist in the ordered match graph $G$. If no such edge was found, start with $2'$ and swipe over $3'$, $4'$ etc, and so on, to the point that the first new edge $i'$ to $j'$ is found. Without loss of generality assume $i'$ is weakly more valuable than $j'$, so $i'<j'$. Since the link between $i'$ and $j'$ does not exist in $G$, at least one of \eqref{eq: cond 1 join} or \eqref{eq: cond 2 join} did not hold when $i'$ was swiping $j'$ in the ordered match algorithm. We show in the first case $i'$ strongly prefers $G$ to $G'$ and in the second case $j'$ strongly prefers $G$ to $G'$.

\begin{enumerate}
\item In the first case that \eqref{eq: cond 1 join} does not hold, we start from $S'_{i'}$ and construct $S_{i'}$ through first a replacement algorithm forming $S^{rep}_{i'}$ and then a replacement algorithm forming $S_{i'}$ from $S^{rep}_{i'}$. Both of these stages are run through steps that only weakly increase the $i'$ preference. By considering different cases we show that at least one of them strongly increases the $i'$ preference. We use the following corollary for our proof here.

\begin{corollary}\label{cor: case 1 replacement} $i'$ strongly prefers all nodes in $S_{i'}\backslash (S_{i'}\cap S'_{i'})$ to $j'$. Also $i'$ weakly prefers $j'$ to all nodes in $S'_{i'}\backslash (S_{i'}\cap S'_{i})$.
\end{corollary} 
The corollary holds since in the case where \eqref{eq: cond 1 join} does not hold, $i'$ is not connected to $j'$ in $G$ or any node less valuable than $j'$. Therefore $i'$ strongly prefers all nodes in $S_{i'}$, and consequently all nodes in $S_{i'}\backslash (S_{i'}\cap S'_{i'})$, to $j'$. On the other hand, since $i'j'$ was the first edge found in the ordered search algorithm, $j'$ is weakly preferred by $i'$ to all nodes in $S_{i'}\backslash (S_{i'}\cap S_{i})$.

The replacement algorithm forms $S^{rep}_{i'}$ as follows:
\begin{itemize}
\item For a number of repetitions equal to $\min\{|S_{i'}\backslash (S'_{i'}\cap S_{i'})|, |S'_{i'}\backslash (S'_{i'}\cap S_{i'})|\}$ do
\begin{enumerate}
    \item  Consider set $S'_{i'}$ and replace an arbitrary node in $S'_{i'}\backslash (S_{i'}\cap S'_{i'})$ by an arbitrary node in $S_{i'}\backslash (S_{i'}\cap S'_{i'})$ to form $S''_{i'}$.
    \item  Replace $S'_{i'}$ with $S''_{i'}$ and repeat the step (a).
\end{enumerate} 
\end{itemize}
We prove that $S_{i'}\succcurlyeq_{i'} S^{rep}_{i'}$. To this end, consider two cases. In case one, $|S_{i'}\backslash (S_{i'}\cap S_{i})|\ge |S'_{i'}\backslash (S'_{i'}\cap S_{i})|$ which means $S^{rep}_{i'}\subset S_{i'}$ and consequently from lemma \ref{lemma: ord match remove edge}, $S_{i'}\succcurlyeq_{i'} S^{rep}_{i'}$. In case two, $|S_{i'}\backslash (S'_{i'}\cap S_{i'})|< |S'_{i'}\backslash (S'_{i'}\cap S_{i'})|$ which means $S_{i'}\subset S^{rep}_{i'}$ in a way that $i'$ prefers all nodes in $S^{rep}_{i'}\S_{i'}$ less than or equal to $j'$. Consequently from Corollary \ref{coro: augmenting}, $S_{i'}\succ_{i'} S^{rep}_{i'}$. Next, in case where $\min\{|S_{i'}\backslash (S'_{i'}\cap S_{i'})|, |S'_{i'}\backslash (S'_{i'}\cap S_{i'})|\}>0$,  due to Corollary \ref{cor: case 1 replacement}, $S^{rep}_{i'}\succ_{i'} S'_{i'}$ and therefore $S_{i'}\succ_{i'} S'_{i'}$. Finally if $\min\{|S_{i'}\backslash (S'_{i'}\cap S_{i'})|, |S'_{i'}\backslash (S'_{i'}\cap S_{i'})|\}=0$, considering that $j'\in |S'_{i'}\backslash (S'_{i'}\cap S_{i'})|$, it should be that $|S_{i'}\backslash (S'_{i'}\cap S_{i'})|=0$ which means $S_{i'}\in S'_{i'}$ in a way that $S'_{i'}\backslash S'_{i'}\cap S_{i'}$ includes nodes that are preferred by $i'$ less than or equal to $j'$. Furthermore in this case $S_{i'}^{rep}=S'_i$. Therefore, from Corollary \ref{coro: augmenting}, $S_{i'}\succ_{i'} S'_{i'}$.

\item In the second case that \eqref{eq: cond 2 join} does not hold, we start from $S'_{j'}$ and construct $S_{j'}$ through the same replacement algorithm as in the first case such that the preference of $j'$ is strongly increased. We skip repetition of the details for this case.
\end{enumerate}

The ordered match algorithm finishes in time $O(N^2)$ since each of the $N$ agents swipes on average over $N/2$ other agents.
\end{proof}

\begin{remark} The ordered match results in a unique outcome.
\end{remark}

\begin{remark} The outcome of the ordered match algorithm is stable and therefore Pareto optimal. However it is not necessarily social-welfare maximizing. Consider simple case that there are two agents $1$ and $2$ with the following utility functions: $u_1(\{1\})=1, u_1(\{1,2\})=0$, $u_2(\{1,2\})=10, u_2(\{2\})=0$. Under the ordered match agent $1$ does not share with agent $2$ since his utility by sharing will decrease. However, the social-welfare-maximizing outcome is for agents $1$ and $2$ to share data.
\end{remark}

Proposition \ref{prop: core} does not hold when either top agent or limited complementarity assumption is relaxed. To see this consider the following examples. 
\begin{example}[No stable outcome without top agent property]\label{ex: limited complementarity counterexample} Consider three agents. Agent $1$ has preferences over $S_1$ defined by $\{1,3\}\succ_{1} \{1,2\}\succ_{1} \{1\} \succ_{1} \{1,2,3\}$; agent $2$ has preferences over $S_2$ as $\{1,2,3\}\succ_2 \{2\} \succ_2 \{2,3\}$; agent 3 has preferences over $S_3$ as $\{1,2,3\}\succ_3 \{3,2\}\succ_3 \{3\} \succ_3 \{3,1\}$. Here the top agent property does not hold. The game does not have a stable outcome.
\end{example}



\section{Unidirectional Sharing: Competitive Prices}\label{sec: comp}

In this Section, we study the case where data is traded with money. To this end we first extend the model in Section \ref{sec: bidirectional sharing} to include utility functions and allow money transfer. Following \cite{shapley1974cores}, we investigate existence of competitive prices that can realize stable outcomes. We will present structures on the utility function that guarantee existence of such prices.

\subsection{Model}
We update the model used in Section \ref{sec: bidirectional sharing}. We assume the agents' preferences from each sharing graph is measurable with a utility function. We consider the following structures for this utility function. 

The sharing graph for the unidirectional sharing will be a directed graph $G$, where a directed edge from $i$ to $j$ indicates agent $i$ sharing his data with agent $j$.

\begin{assumption}[Utility structure]\label{def: util struc}
Agents' preferences over outcomes $G$ can be represented with a total utility function. Agent $i$'s total utility function is $V_i: \mathcal{G}\rightarrow {R}$. Moreover, $V_i$ is decomposed as a utility function $U_i(S_i^I)$ minus a cost function $C_i(S_i^O)$:
\begin{align}
    V_i(G):=U_i(S_i^I)-C_i(S_i^O),
\end{align}
where $S_i^I$ is the set of agents that agent $i$ gains access to their data, and $S_i^O$ is the set of agents that gain access to agent $i$'s data.
\end{assumption}

We now leverage the particular features of data for machine learning algorithms to put structure on the utility functions and the cost function.


\begin{assumption}[Additive Cost Structure]\label{ass: cost structure additive}
Agent $n$'s cost is additive with respect to the agents who access his data:
\begin{align}\label{eq: cost structure additive}
    C_i(S^O_i)=\sum_{j\in S^O_i} c_i(j).
\end{align}
\end{assumption}

\begin{remark}
The cost structure in \eqref{eq: cost structure additive} is additively decomposable. It can represent privacy cost as well as peer-to-peer communication costs. However it does not include one-time communication with an intermediary which can then send the data to the others. It also does not model differential privacy costs where there is a privacy cost per inquiry of agent $j$ with respect to agent $i$'s data. Extension of the model to incorporate differential privacy is considered in Section \ref{sec: dif priv}. 
\end{remark}

\begin{assumption}
Agents can have monetary transfer. We denote the payment by/charge to agent $i$ with $t_i$, and its quasilinear utility function by
\begin{align}\label{eq: total utility}
  V_i(G)-t_i:=U_i(S^I_i)-C_i(S^O_i)-t_i.
\end{align}
\end{assumption}

\subsection{Existence of competitive prices}
\begin{proposition}\label{prop: comp}
If costs are additively decomposable, see (\ref{eq: cost structure additive}) in Assumption \ref{ass: cost structure additive}, then for a unidirectional sharing graph (i) there exist competitive prices, and (ii) every social-welfare-maximizing outcome can be implemented with a competitive price, such that it is also a stable outcome.
\end{proposition}

\begin{proof} We prove by construction. (i) Consider the set of prices where agent $i$ charges a price of $c_i(j)$ for selling his data to agent $j$. We show these prices are competitive (form a competitive equilibrium). The total payment of agent $i$ will be 
\begin{align}
    t_i=\sum_{j\in S^I_j} c_j(i)-\sum_{j\in S^O_j} c_i(j).
\end{align}
We determine the demand and supply of each agent $i$. Under Assumption \ref{ass: cost structure additive},
\begin{align}\label{eq: utility competitive prices}
    U_i(S^I_i)-C_i(S^O_i)-t_i&=U_i(S^I_i)-\sum_{j\in S^O_i} c_i(j)-\sum_{j\in S^I_j} c_j(i)+\sum_{j\in S^O_j} c_i(j)\\
    &=U_i(S^I_i)-\sum_{j\in S^I_j} c_j(i).
\end{align}
Agent $i$' demand, which is the optimal subset of agents to buy data from denoted by ${S^I_i}^*$, is calculated by 
\begin{align}\label{eq: comp indiv}
    {S^I_i}^*\in\arg\max_{S^I_i}U_i(S^I_i)-\sum_{j\in S^I_i} c_j(i).
\end{align}
Agent $i$ is indifferent on the amount of supply, which is the set of agents agent $i$ shares its data with; this is because (\ref{eq: utility competitive prices}) is independent of $S^O_i$. Consequently, for any demand set $S^I_j$, $j\in \{1,2,...,N\}$, the optimal supply of agent $i$ can be determined uniquely in a way to meet the demand and clear the market:
\begin{align}
    {S^O_i}^*=\{j: j\in \{1,2,...,N\}, i\in {S^I_j}^*\}.
\end{align}

(ii) We prove that every social-welfare-maximizing outcome can be implemented with the competitive prices we proposed above. Start with a social-welfare-maximizing outcome and decompose it as follows:
\begin{align}
    \max_{G\in \mathcal{G}} \sum_{i=1}^{N} U_i(S^I_i)-C_i(S^O_i)&= \sum_{i=1}^{N}\max_{S_i^I\subset \{1,2,...,N\}}\bigg( U_i(S^I_i)-\sum_{j\in S^I_i}c_j(i)\bigg)\label{eq: comp all}\\
    &= \sum_{i=1}^{N}\max_{S_i^I\subset \{1,2,...,N\}, S_i^O\subset \{1,2,...,N\}}\bigg( U_i(S^I_i)-C_i(S^O_i)-t_i\bigg),
    \label{eq: comp all}
\end{align}
where the decomposition in (\ref{eq: comp all}) matches that of (\ref{eq: comp indiv}). This shows if every agent optimizes (\ref{eq: comp indiv}) individually, then the solution will be social welfare maximizing. To see under the proposed competitive prices, the social welfare maximizing outcomes are stable, note that from (\ref{eq: comp indiv}) every agent is gaining its maximum utility at the outcome of the social-welfare-maximizing outcome  compared to any other coalition.
\end{proof}
\begin{remark}
The results in Proposition \ref{prop: comp} are based on two features in this problem which are specific to the data sharing economy: first, the sharing costs are decomposable as a sum of the sharing cost between each pair of individuals according to Eq. (\ref{eq: cost structure additive}) in Assumption \ref{ass: cost structure additive}. Second, data can be replicated freely so there is no constraint on number of agents who buy data from an agent. This last assumption makes it possible to define market-clearing supply amounts. 
\end{remark}

\begin{proposition}[Space of the Competitive Prices] Competitive prices of the unidirecitonal sharing game are not unique to those presented in Proposition \ref{prop: comp}. For example, fixing all other prices to be $c_i(j)$, $i\neq m$, agent $m$ can set its prices for agent $j$ to be any value in $[c_m(j), p^{\max}_m(j)]$ where $p^{\max}_m(j)$ is a price such that the solutions to the following optimization remains the same as that in (\ref{eq: comp indiv}) for all $i\in \{1,2,...,N\}$:
\begin{align}
    \arg\max_{S^I_i}U_i(S^I_i)-I(m\in S_n^I)p^{\max}_m(i)-\sum_{j\in S^I_n, j\neq m } c_j(i).
\end{align}
\end{proposition}

\section{Unilateral Sharing with Private Types: Mixed-VCG Mechanism}\label{sec: mech}

In this section we extend the model by assuming agents belong in type space $\Theta$, determining their utility from sharing. We further assume types are private information known only to the agents themselves.  There also exists a social planner intermediary who runs a market for coordinating data sharing. We look for market mechanisms that have the following features: they implement approximate social-welfare-maximizing outcomes at dominant equilibrium, they are individually rational, and they are budget balanced. To this end, we propose a variation of the VCG mechanisms---mixed VCG---which uses a mix of data distortion (referred to as ``data money'') and monetary payment to achieve budget balance. We will then investigate how to address communication/message complexity of the VCG mechanisms for data sharing, using specific features of data.


\subsection{Model}
We extend the model in Section \ref{sec: comp}.

\begin{assumption}
Agents are from type space $\Theta$ with distribution $f_{\theta}$. Types are private information.
\end{assumption}

Note that type $\theta_n$ is a determinant of agent $n$'s utility and cost functions. It does not involve  its private data $D_n$.

\begin{assumption}
There exists an intermediary for running the market. The intermediary is a social planner.
\end{assumption}

Furthermore, we extend the game by making the following two assumptions which capture unique features of data economy. We discuss at the end of this section how these assumptions can be relaxed. 
\begin{assumption}\label{ass: zero cost data distortion}[Zero Cost Data Distortion]
    The intermediary has the capacity to distort the quality of data passed to each agent downward or upward, at zero cost to the intermediary.
\end{assumption}
There are several ways the quality of data can be distorted. For example, downward distortion could be achieved by removing part of one agent's data when passed to another, and upward distortion could be achieved by adding new data passed to an agent (perhaps by simulating new data from a model). Alternatively, the intermediary can add noise to data, or to the machine learning model passed to each agent. Later in this section, we will relax the assumption of zero cost distortion to the intermediary.

\begin{assumption}\label{ass: isolated distortion impact}[Isolated Impact of Data Distortion]
We assume that data rationing by the intermediary for data shared with one agent does not change the total utility of the other agents. 
\end{assumption}

Assumption \ref{ass: zero cost data distortion} extends the outcomes of the game from the graph $G$ in Section \ref{sec: bidirectional sharing} to a weighted directed graph $\hat{G}$ defined below.

Assumption \ref{ass: isolated distortion impact} implies that other agents do not incur a change in their utility or costs from data sharing, such as privacy cost, due to rationing on distortion of another agent's data. This allows for isolating the impact of rationing data passed to an agent only to itself.

\begin{definition}
The outcome of a unidirectional sharing game where the intermediary can change quality of the data is a weighted directed graph $\hat{G}$. The weight of each edge from $i$ to $j$ is denoted by $w_{i,j} \in [0,\infty)$. 
\end{definition}
One way to interpret the weights is as the fraction of the data from $i$ passed by the intermediary to $j$.

We also extend the utility of the agents to be explicitly a function of their type and the weighted directed graph $\hat{G}$:
\begin{align}\label{eq: V vcg}
    V(\hat{G}, \theta_i)-t_i=U(\hat{S}^I_i, \theta_i)-C(\hat{S}^O_i, \theta_n)-t_i.
\end{align} 
Similar to (\ref{eq: total utility}), $U$ is the utility from gaining access to other agents' data, $C$ is the cost of sharing, and $t$ is the monetary payment.


Before introducing our mechanism we provide the following definition which helps in studying the properties of our mechanism.

\begin{definition}
Agents' social welfare is 
\begin{align}\label{eq: SW A}
    SW^A(\hat{G}):=\sum_{i\in \{1,2,..., N\}} \bigg(V(\hat{G}, \theta_i)-t_i\bigg).
\end{align}
Total social welfare is the sum of utilities of all agents plus the payments to the intermediary:
\begin{align}\label{eq: SW T}
    SW^T(\hat{G}):=&SW^A(\hat{G})+\sum_{i\in \{1,2,...,N\}} t_i\\
    =&\sum_{i\in \{1,2,..., N\}} V(\hat{G}, \theta_i).
    \label{eq: agents allocation welfare}
\end{align}
\end{definition}
\begin{remark}  $SW^T$ in (\ref{eq: agents allocation welfare}) is agents' utility only from allocations (including costs imposed to the intermediary for data distortion) excluding payments. The standard classical VCG mechanism implements maximum total social welfare $SW^T$ (often referred to simply as social welfare). VCG does not implement maximum $SW^A$.
\end{remark}

\subsection{Mixed-VCG mechanism}
We introduce a mechanism that we refer to as \emph{Mixed-VCG} to achieve approximately optimal $SW^A$, truthful reporting as a dominant strategy, budget balance and ex-post individual rationality Moreover, the mechanism will have the same $SW^T$ as standard VCG.

\begin{definition}[Mixed-VCG] Agents report their type to the intermediary; agent $i$'s reported type is denoted by $\hat{\theta}_i$s. The intermediary then outputs a weighted directed sharing graph $\hat{G}$ and a payment vector $t$ by taking the following steps.

First the intermediary determines a sharing graph $G^*$ similar to standard VCG and without any distortions,
\begin{align}\label{eq: standard vcg}
    {G}^*&=\arg\max_{{G}\in \mathcal{G}} \sum_{i\in \{1,2,...,N\}} V({G}, {\theta}_i).
\end{align}
 determines $\tilde{t}_i$ as
\begin{align}
   \tilde{t}_i= \sum_{j\neq i} V({G}_{-i}^*, {\theta}_j)-\sum_{j\neq i}V({G}^*, {\theta}_j),
\end{align}
where 
\begin{align}
    {G}_{-i}^*=\arg\max_{{G}\in \mathcal{G}}\sum_{j\neq i} V({G}, {\theta}_i).
\end{align}
and calculates the total budget surplus or deficit of the standard VCG:
\begin{align}\label{eq: Delta}
    \Delta=\sum_{i} \tilde{t}_i.
\end{align}
The intermediary then determines each agent $n$'s real payment $t_n$ and data payment $t^d_n$ (which will be translated into distortion to agent $n$'s data) such that  
\begin{align}\label{eq: mixed payment 1}
     t_n+t^d_n=\tilde{t}_n.
\end{align}
Without loss of generality assume $\tilde{t}_1>\tilde{t}_2>\cdot\cdot\cdot$. If $\Delta>0$---i.e., there is budget surplus in standard VCG---then the intermediary starts from agent $1$ and sets 
\begin{align}\label{eq: td 1}
    t^d_1&=\min\{\tilde{t}_1,\Delta\}\\
    t^{d}_n&=\min\{\tilde{t}_n, \Delta-\sum_{j=1}^{n-1}t^d_1\} \qquad \forall n=2,3,...,N.
    \label{eq: td n}
\end{align}
If $\Delta<0$, i.e., there is a budget deficit, then the intermediary starts from agent $N$ and sets
\begin{align}
    t^d_N&=\max\{\tilde{t}_N,\Delta\}\label{eq: td 1 2}\\
    t^{d}_n&=\max\{\tilde{t}_n, \Delta-\sum_{j=N}^{n+1}t^d_n\} \qquad \forall n=N-1,N-2,...,1.\label{eq: td n 2}
\end{align}
Next the intermediary sets monetary payments as follows:
\begin{align}\label{eq: mixed vcg t}
t_i=\tilde{t}_i-t^d_i \qquad \forall n\in N,
\end{align}
and calculates data distortions, $w^*_{i,j}, \forall i,j$, to form the sharing graph $\hat{G}^*$ so that the following is satisfied:
\begin{align}\label{eq: mixed vcg hat d}
    V(\hat{G}^*, \theta_i)-t_i= V(G^*, \theta_i)- \tilde{t}_i.
\end{align}
\end{definition}

\begin{proposition}
The mixed-VCG mechanism is dominant truthful, budget balanced and ex-post individually rational. It implements agents' social welfare, $SW^A$, whose  suboptimality is exactly equal to the level of budget imbalance of standard VCG. It also implements the same $SW^T$ as standard VCG.  
\end{proposition}%
\begin{proof}
To see individual rationality and dominant truthfulness, it is sufficient to note that for the  vector of reported types $(\hat{\theta}_1, \hat{\theta}_2,..., \hat{\theta}_N)$ based on (\ref{eq: mixed vcg hat d}), all agents gain the same utility for both standard VCG and mixed-VCG. Individual rationality and dominant truthfulness then follow from the fact that these properties hold for standard VCG.

Budget balance is a direct consequence of definition of the monetary payments in mixed-VCG, which is the redistribution of the standard VCG budget imbalance across agents according to \eqref{eq: mixed payment 1}-\eqref{eq: mixed vcg t}. This redistribution ensures $\sum_{i=\{1,2,...,N\}} t_i=0$.

From \eqref{eq: mixed vcg hat d}, $\sum_{i=\{1,2,...,N\}} V(\hat{G}^*, \theta_i)-{t}_i=\sum_{i=\{1,2,...,N\}} V(G^*, \theta_i)-\tilde{t}_i$, so from \eqref{eq: SW A} mixed-VCG achieves the same $SW^A$ as the standard VCG.

To see the difference of the $SW^T$
between mixed-VCG and standard VCG, from \eqref{eq: agents allocation welfare}, \eqref{eq: mixed vcg hat d}, budget balance of the mixed-VCG mechanism which induces $\sum_{i=\{1,2,...,N\}} t_i=0$, and \eqref{eq: Delta},
\begin{align}
SW^T_{mixed-VCG}&=\sum_{\{i=1,2,...,N\}} V(\hat{G}^*, \theta_i)\\
&=\sum_{\{i=1,2,...,N\}} V({G}^*, \theta_i)+\sum_{\{i=1,2,...,N\}}t_i-\sum_{\{i=1,2,...,N\}}\tilde{t}_i)\\
&=\sum_{\{i=1,2,...,N\}} V({G}^*, \theta_i)-\Delta=SW^T_{VCG}-\Delta.
\end{align}

\end{proof}

\begin{remark} The NM-VCG mechanism  relies on isolated impact of data distortion in Assumption \ref{ass: isolated distortion impact}, the fact that the intermediary can distort the data allocations to change an individual's utility without changing the others'. This is a property of data economy and not the case for other commodities (for example house allocation). 
\end{remark}

We now consider relaxing Assumption \ref{ass: zero cost data distortion}. While distorting data quality downward is practically at zero cost, for example by adding noise to the data, an upward distortion may be practically costly to the intermediary, for example through training a model for generating new data. To address this issue, we extend the mixed-VCG mechanism. The intermediary can add a fixed base downward distortion of data in a zero cost way, for example fixed noise added to data. Then instead of adding new data, the intermediary just lifts the fixed noise which will be at negligible cost in practice. We call this mechanism \emph{Distorted mixed-VCG} (D-mixed-VCG). Formally consider $w^0_{n,m}$ to be the chosen ex-ante weights placed on the edges. Denote the class of all such sharing graphs as $\hat{\mathcal{G}}_0$ (the set of graphs that $i$ either does not pass his data to $j$ or gives a distorted version at weight $w_{i,j}$. Then intermediary determines a sharing graph $\hat{G}^*_0$ according to
\begin{align}\label{eq: standard vcg 2}
    \hat{G}_0^*&=\arg\max_{\hat{G}_0\in \hat{\mathcal{G}}_0} \sum_{i\in \{1,2,...,N\}} V(\hat{G}_0, {\theta}_i),
\end{align}
and determines $\tilde{t}_i$ as
\begin{align}
   \tilde{t}_i= \sum_{j\neq i} V(\hat{G}_{0,-i}^*, {\theta}_j)-\sum_{j\neq i}V(\hat{G}_0^*, {\theta}_j),
\end{align}
where 
\begin{align}
    \hat{G}_{0,-i}^*=\arg\max_{\hat{G}_0\in \mathcal{\hat{G}}_0} \sum_{j\neq i} V(\hat{G}_{0,-i}^*, {\theta}_i).
\end{align}
the rest of the mechanism is analogous to the mixed VCG mechanism per \eqref{eq: Delta}-\eqref{eq: mixed vcg hat d}.

D-mixed-VCG implements lower $SW^T$ and $SW^A$ compared to mixed-VCG because it adds distortion to the data. The proof is straightforward, and we leave it to the reader.


\begin{remark}
We now discuss the communication complexity of the mixed-VCG mechanism for data sharing by using specific features of data as a commodity. From statistical machine learning results, the agents' utility and costs from sharing data, $U$ and $C$, are of certain structural forms; for example, for several machine learning methods, the marginal value of the new data for learning a model reduces at a rate of $\frac{1}{\sqrt{D}}$ where $D$ is the size of the entire data available to the agent. Using such structural properties allows reducing the complexity of the type space, hence the communication complexity.
\end{remark}

\section{Extensions to Differential Privacy}\label{sec: dif priv}
In this section we briefly show how the model and results can be adapted to a more complicated setting, by bringing differential privacy into the picture. We relax full reporting of data in Assumption \ref{ass: full sharing}, and assume agents make queries over each others' data, where for each query there is a differential privacy cost for the data owner. In this setup, the unit of commodity offered by each agent $n$ is the outcome of a query over its data, rather than its whole raw data. We omit the proofs of this section.

We update the model by extending the results in Section \ref{sec: bidirectional sharing} to \ref{sec: mech}. First, the sharing graph $G$ in Definition \ref{def: sharing graph} is extended with integer weights on the edges $\tilde{w}_{i,j}$ which reflect the number of queries agent $i$ runs on agent $j$'s data. We refer to this graph by $\tilde{G}$. The cost and utility functions in Assumptions \ref{ass: cost structure additive} are extended to be a function of graph $\tilde{S}^O_i$ and $\tilde{S}^I_i$,
$V_i(\tilde{G})=U_i(\tilde{S}^O_i)-C_i(\tilde{S}^O_i)$, where $\tilde{S}^I_i=\{(j,\tilde{w}_{i,j}): j\in N\}$, $\tilde{S}^O_i=\{(j,\tilde{w}_{j,i}): j\in N\}$
Given these updates to the model, together with Proposition \ref{prop: core}, we can show that a core stable outcome exists and can be found by updating the ordered match algorithm as follows: every time an agent $n$ swipes an agent $\ell$, it also proposes a set of sharing weights $\{(\tilde{w}_{n,\ell},\tilde{w}_{\ell, n})\}$ that are acceptable to $n$ (the benefits of sharing are greater than costs). Agent $\ell$ then either accepts to share by choosing from this set or rejects the sharing of data. 

Next consider a cost structure in which each agent $i$ has a differential privacy cost which, along the line of (\ref{eq: cost structure additive}) is additively decomposable with respect to the agents he is sharing data with:
\begin{align}
    C_i(\tilde{S}^O_i)=\sum_{j\in \tilde{S}^O_i} \tilde{w}_{i,j} c_i(j).
\end{align}
Then, along the lines of Proposition \ref{prop: comp}, there exists a set of competitive prices where each agent $i$ charges agent $j$ an amount of $c_i(j)$ per inquiry. Moreover, each other agent $j$ is allowed to purchase multiple inquiries from agent $i$. 

Finally, the VCG mechanism in Section \ref{sec: mech} can be updating by considering doubly weighted directed graphs $\hat{\tilde{G}}$ with weights on each edge from $i$ to $j$ being $\tilde{w}_{i,j}$ the number of inquiries of $j$ on $i$'s data, and $w_{i,j}$ the distortion to quality of each inquiry. Then (\ref{eq: V vcg}) can be extended as
\begin{align}\label{eq: V vcg}
    V(\hat{\tilde{G}}, \theta_i)-t_i=U(\hat{\tilde{S}}^I_i, \theta_i)-C(\hat{\tilde{S}}^O_i, \theta_i)-t_i.
\end{align}
The mixed-VCG mechanism is updated to include weights $\tilde{w}_{ij}$ determined by the intermediary.

\section{Conclusions}
\label{sec: conclusion}

We have studied data sharing markets for both bidirectional (data-data exchange) and unidirectional (data-money exchange) sharing. We proved existence of strongly stable outcomes for bilateral sharing via formulating the problem as a network formation game. We achieved this with limited complementarity and without posing utility structure on agents' preferences. For unilateral sharing, we constructed competitive prices that implement socially optimal outcomes, and in the case of agents' private information of their types, we also presented budget balanced mechanisms which truthfully implement approximately optimal outcomes.




\bibliographystyle{apalike}
\bibliography{main}


\end{document}